\documentclass[11pt]{llncs}
\usepackage[smallbib]{}
\usepackage{makeidx}
\usepackage{amsfonts}
\usepackage{verbatim}
\usepackage{color}              % Need the color package
\usepackage{epsfig}
\usepackage{psfig}           % Capability to place postscript drawings
\usepackage{fullpage}

\frontmatter          % for the preliminaries

\usepackage{amssymb,amsfonts,amsthm}
\usepackage{amsmath,epsfig,url,graphicx,xspace}
\usepackage{algorithm}
\usepackage{algorithmic}
\usepackage{fullpage}
\newcommand{\indsets}{{\mathcal{I}}}

\begin{document}

%\title{Extensions to the Buyback Problem - Online algorithms with cancellations}
\title{Buyback Problem - Approximate matroid intersection with cancellation costs}

\author{Ashwinkumar B.V.}
\institute{Cornell University, Ithaca, NY \\ \email{ashwin85@cs.cornell.edu}}

\maketitle

\begin{abstract}
In the buyback problem, an algorithm observes a sequence of bids and must decide whether to accept each bid at the moment it arrives, subject to some constraints on the set of accepted bids. Decisions to reject bids are irrevocable, whereas decisions to accept bids may be canceled at a cost that is a fixed fraction of the bid value. Previous to our work, deterministic and randomized algorithms were known when the constraint is a matroid constraint. We extend this and give a deterministic algorithm for the case when the constraint is an intersection of $k$ matroid constraints. We further prove a matching lower bound on the competitive ratio for this problem and extend our results to arbitrary downward closed set systems. This problem has applications to banner advertisement, semi-streaming, routing, load balancing and other problems where preemption or cancellation of previous allocations is allowed.
\end{abstract}

\section{Introduction}
\label{sec:intro}
Consider the online problem of resource allocation in which preemption is allowed. This kind of problem has been heavily studied in a wide variety of settings which range from advertisement allocations to routing to load balancing. In online weighted resource allocation without preemption we cannot get any non-trivial worst case guarantee on the sum of weights allocated. Consider the simplest problem of choosing the maximum number in a sequence. Any deterministic or randomized algorithm cannot have a constant competitive ratio up to any factor. Usually this impossibility is circumvented in the literature by placing some restrictions. This could be by allowing the input to be either a random permutation \cite{BIK07,BIKK07,Dyn63,KP-icalp09,K-soda05} or drawn iid from some probability distribution \cite{HKS-aaai07,CHMS-stoc10}. Other approach \cite{BHK-ec09,CFMP-soda09,FKMMP-wine09,AA-soda99,CS-stoc95,ABM-spaa03,GG-info92,GGKMY-jalgorithms97} which does not relax any conditions on the input assumes that either preemption is allowed or preemption with a penalty is allowed and gives guarantee for every input. In this paper we study this kind of relaxation. 

%We will assume familiarity with matroids \cite{oxley} and competitive analysis \cite{BorodinElYaniv}.

Consider the following generic problem. There is a set system $\mathcal{I}$ (downward closed) for the ground set $\mathcal{E}$. Elements from $\mathcal{E}$ are presented to the algorithm in a sequential manner. Each element $e_i$ is also associated with a utility $w_{e_i}$. When element $e_i$ is presented to the algorithm it must be accepted or rejected immediately. When $e_i$ is accepted the algorithm could cancel (preempt) some of the previously accepted elements. If $S$ denotes the set currently accepted, then the constraint for the algorithm is to have $S\in \mathcal{I}$. The utility of the algorithm is the utility of the accepted elements minus the penalty paid to the canceled elements. All the canceled elements are paid a penalty proportional to their corresponding utility. We present some applications of this generic problem below.

\subsubsection{Banner advertisement} The buyback problem was first defined and studied in \cite{BHK-ec09,CFMP-soda09}. Specifically they give deterministic algorithms for the case when $\mathcal{I}$ is a matroid. This was later extended by \cite{AK-wine09} which gave a randomized algorithm with better competitive ratio. Consider an advertisement system for a single advertisement slot. In certain systems bidding for this slot starts well in advance. In such a system bidders come and bid in an online manner. The system accepts the bids or rejects them immediately. The system could later accept much higher bids and reject previously accepted ones. But this causes a loss for the previously accepted bidders. Hence the system pays them back with a penalty. The work in \cite{BHK-ec09,CFMP-soda09,AK-wine09} also generalizes to much more general systems where the accepted bids can form a matroid. They leave open the question of finding algorithms for more general constraints. One of the key constraints not modeled by this is when each bidder desires a single item among a set of items, i.e when $\mathcal{I}$ is a valid matching in a bipartite graph. Our result in section \ref{s:preliminaries} solves this as well as a generalization to $k$ arbitrary matroid constraints and our result in section \ref{s:extensions} generalizes this to any downward closed set system. We also prove matching lower bounds in section \ref{s:lowerbound}. It is important to note that we ignore the incentives throughout our work and assume that bids/values are truthfully reported. 

\subsubsection{Free Disposal} Consider the problem of online ad allocation with free disposal studied in \cite{FKMMP-wine09}. Here we have a set of advertisers $A$ known in advance together with an integer impression contract $n(a)$ for each advertiser $a\in A$. $n(a)$ denotes the maximum number of impressions for which advertiser can be charged. When an impression $i\in I$ arrives the utility of this impression $w_{i_a}$ for every advertiser $a$ is also revealed. The final utility of the algorithm is the total amount charged to each user. We note that there is a very straight forward reduction from this problem to the problem we define, specifically from the case when $k=2$ and zero penalty for preemption($f=0$). Our algorithm gives a 5.828 competitive ratio. \cite{FKMMP-wine09} gives an unconditional $2$ competitive and conditional $e/(e-1)$ competitive algorithm. Note that our algorithm is tight for the generic problem defined due to the lower bound shown in section \ref{s:lowerbound}. \cite{FKMMP-wine09} is able to give better competitive ratio as this problem is more restrictive than the generic problem we define.

\subsubsection{Semi Streaming} Recently few papers \cite{G-approx05,FKMSZJ-soda05,DFRA-soda06,CM-pods05,SGPR-pods08,BGR-tcs08,FKMSZ-icalp04} have studied graph problems in the semi-streaming model. Consider the problem of finding a weighted matching in a stream which uses $\tilde{O}(n)$ memory, $O(1)$ update time and 1 pass. \cite{FKMSZ-icalp04} introduced this problem and gave a factor 6 approximation. \cite{G-approx05} improved this to $5.828$ approximate semi-streaming algorithm. Both of these algorithms are characterized by the fact that they always maintain a valid matching. Hence our lower bound in section \ref{s:lowerbound} proves that among the class of algorithms which maintain only edges of a valid matching no algorithm can achieve better than approximation ratio $5.828$.

\subsubsection{Routing with preemption and Load Balancing} There has been a huge literature \cite{AA-soda99,CS-stoc95,ABM-spaa03,GG-info92,GGKMY-jalgorithms97} on Routing in Networks and Load Balancing where preemption of previously allocated resources in allowed. Many of these results studied can be very succinctly generalized by the following problem.

Consider an arbitrary downward closed set system $\mathcal{I}$. Elements $e_i\in \mathcal{I}$ are presented to the algorithm along with its weight and the sets to which it belongs to. The algorithm has to either accept or reject the element immediately. An accepted element can be preempted or canceled later but canceled/rejected element cannot be taken later.

This problem is precisely the problem we study with the restriction of penalty for cancellation being 0. As we will show in section \ref{s:extensions}, such a generic problem does not have any algorithm with competitive ratio better than $n-1$ even when penalty is 0. Also a trivial extension of the algorithm from \cite{BHK-ec09} will give a competitive ratio of $(n-1)(1+2f+2\sqrt{f(1+f)})$. One should note that the papers which study routing with preemption and load balancing are able to achieve better competitive ratio by exploiting some additional structure in the problem.

\subsubsection{} The offline problem of matroid intersection was introduced and studied in \cite{J-computing76,KH-ann78}. The problem has also been studied under more general submodular utility functions in the CS theory/OR literature. Some of the recent papers in this direction are \cite{UMV-focs07,NWF-math78,RS-gecco07,LMNS-stoc09}. 

\subsubsection{Our Contributions} The algorithm we propose is a greedy type algorithm which is a natural generalization of the algorithm from \cite{BHK-ec09,CFMP-soda09}. While the algorithm is simple to state the analysis turns out to be much harder. There are two key technical contributions in this paper.
\begin{enumerate}
\item The main technical hurdle in analyzing the algorithm comes from bounding the utility of the final accepted set $S$ as compared to the optimal set $OPT$(Lemma \ref{mainlemma}). For this we develop a new novel type of charging scheme. This charging scheme is also aided by a graph construction which could be of independent interest in discrete optimization.
\item We also give optimal lower bounds. The main technical hurdle in this case is identifying the input on which the algorithm has worst case behavior. There is a key difference in the lower bound for $k=1$ case  (given in \cite{BHK-ec09}) and $k=2$ case. The general $k$ is quite similar to $k=2$. The matroid used in the lower bound for $k=1$ case in \cite{BHK-ec09} is just a uniform matroid of rank $1$. For the case $k=2$ we need a intersection of two partition matroids (i.e matching) which has a lot more structure. This is used to derive an infinite sequence $\{x_1,x_2,...\}$ with certain properties. 
\end{enumerate}
We also note here the portions of the paper which follow more easily from \cite{BHK-ec09}. In the case of analyzing the algorithm the part which bounds the penalty uses a geometric series argument and is quite akin to \cite{BHK-ec09}. Similarly in the lower bound once the sequence $\{x_1,x_2,...\}$ along with its properties are identified the rest of the analysis follows more easily.

\section{Preliminaries}\label{s:preliminaries}
First we define the problem formally 
\subsection{Model}
Consider a ground set of Elements $\mathcal{E}=\{e_1,e_2,\ldots,e_n\}$. Let $\mathcal{M}_1,\ldots,\mathcal{M}_k$\footnote{A matroid $\mathcal{M}_i=(\mathcal{E},\mathcal{I}_i)$ is constructed from a ground set $\mathcal{E}\neq \phi$ and a nonempty family of subsets of $\mathcal{E}$, called the independent subsets of $\mathcal{E}$, such that if $B\in \mathcal{I}_i$ and $A\subseteq B$ then $A\in \mathcal{I}_i$ ($\mathcal{I}_i$ is hereditary). Additionally, if $A,B\in \mathcal{I}_i$ and $|A|<|B|$, then there is some element $x\in B-A$ such that $A\cup \{x\}\in \mathcal{I}_i$ (exchange property).} be $k$ arbitrary matroid constraints on $\mathcal{E}$. Let the corresponding Independent sets be $\mathcal{I}_1,\ldots,\mathcal{I}_k$ and let $\mathcal{I}=\cap_{j=1}^k \mathcal{I}_j$. We define the online problem with the following constraints.
\begin{enumerate}
\item The elements of $\mathcal{E}$ are presented to the algorithm in some arbitrary order. The value $w_{e_i}$ of element $e_i$ and the matroid constraints it is involved with are revealed to the algorithm when the element is presented to it.
\item When element $e_i$ is presented it must be accepted or rejected immediately. Additionally it could be canceled at a later point in time. When an element is canceled the algorithm must pay a penalty $f\cdot w_{e_i}$ where $f$ is a constant. (Note the difference between reject and cancel)
\item Let $\mathcal{A}$ be the set of elements accepted and $\mathcal{R}$ be the set of elements accepted and later canceled. Then the utility of the algorithm is defined as $\sum_{e\in\mathcal{A}}w_e-(1+f)\sum_{e\in\mathcal{R}}w_e$. Note that all elements in $\mathcal{R}$ are also counted in $\mathcal{A}$. Moreover the currently maintained set $S=\mathcal{A}-\mathcal{R}$ must be an independent set.i.e $S\in \mathcal{I}$.
\end{enumerate}
Here we desire to find a competitive online algorithm with the above constraints.

\subsection{Algorithm}
\begin{comment}
\begin{algorithm}[h]
\caption{Online Matroid Intersection}
\label{alg:matroid}
\begin{algorithmic}[1]
\STATE Initialize $S = \emptyset.$
\FORALL{elements $e_i$, in order of arrival,}
\IF{$S\cup \{e_i\}\in \indsets$}
\STATE  $S=S\cup \{e_i\}$  \label{line:accept} 
\ELSE
\FORALL{$1\leq j\leq k$}
\IF{$S\cup \{e_i\}\notin \indsets_j$}
\STATE $e_{i_j}$ be the element of smallest value such that $S\cup \{e_i\}\setminus \{e_{i_j}\}\in \indsets_j$
\ELSE
\STATE $e_{i_j}=NULL$.
\ENDIF
\ENDFOR
\STATE Let $C_{e_i}=\cup_{j=1}^k\{e_{i_j}\}$
\IF{$w_{e_i}\geq r\cdot (\sum_{j=1}^k w_{e_{i_j}})$}\label{line:algcompare}
\STATE $S=S\cup \{e_i\} \backslash \cup_{j=1}^k \{e_{i_j}\}$
\ENDIF
\ENDIF
\ENDFOR
\end{algorithmic}
\end{algorithm}
\end{comment}

\begin{figure}[t]
\begin{minipage}[t]{0.5\linewidth}
\centering
%\begin{algorithm}
%\caption{Random Filtering Algorithm $\Filter{\ALG}$}
\label{alg:onlinematroidintersection}
{\bf Algorithm 1: Online Matroid Intersection:}
\vspace*{-1mm}
\begin{algorithmic}[1]
\STATE Initialize $S = \emptyset.$
\FORALL{elements $e_i$, in order of arrival,}
\IF{$S\cup \{e_i\}\in \indsets$}
\STATE $S=S\cup \{e_i\}$
\ELSE
\FORALL{$1\leq j\leq k$}
\IF{$S\cup \{e_i\}\notin \indsets_j$}
\STATE $e_{i_j}$ be the element of smallest value such that $S\cup \{e_i\}\setminus \{e_{i_j}\}\in \indsets_j$
\ELSE
\STATE $e_{i_j}=NULL$.
\ENDIF
\ENDFOR
\STATE Let $C_{e_i}=\cup_{j=1}^k\{e_{i_j}\}$
\IF{$w_{e_i}\geq r\cdot (\sum_{j=1}^k w_{e_{i_j}})$} 
\STATE $S=S\cup \{e_i\} \backslash \cup_{j=1}^k \{e_{i_j}\}$
\ENDIF
\ENDIF
\ENDFOR
\end{algorithmic}
\end{minipage}
\hspace{3mm}
\begin{minipage}[t]{0.48\linewidth}
\centering
\vspace*{0.0mm}
{\bf Algorithm 2: Online Matroid Intersection:}
\vspace*{-1mm}
\begin{algorithmic}[1]
\STATE Initialize $S = \emptyset$.
\FORALL{elements $e_i$, in order of arrival,}
\IF{$S\cup \{e_i\}\in \indsets$}
\STATE $S=S\cup \{e_i\}$
\ELSE
\STATE $S'$=Greedy on $S\cup \{e_i\}$ \footnote{Here Greedy means sorting and choosing the max weight elements satisfying the matroid constraints}
\STATE Let $C_{e_i}=S\cup \{e_i\}-S'$ (Rejected Elements)
\IF{$w_{e_i}\geq r\cdot (\sum_{e\in C_{e_i}} w_e)$}
\STATE $S=S'$
\ENDIF
\ENDIF
\ENDFOR
\end{algorithmic}
%\end{algorithm}
\end{minipage}
\caption{Algorithms for $k$ matroids intersection}
\label{fig:algorithms}
\end{figure}

The algorithm is shown in Figure 1 as Algorithm 1. At each step the algorithm maintains an Independent set S. Assume it sees the element $e_i$ at some step. If $S\cup \{e_i\}$ is also an independent set, then it includes $\{e_i\}$ into the current set $S$. Otherwise $S\cup \{e_i\}$ has a circuit in some of the matroids $\indsets_j$. It first finds the minimum value element ($e_{i_j}$) it must remove in set S to make $S\cup \{e_i\}$ an independent set in each of $\indsets_j$. Now suppose $w_{e_i}\geq r\cdot (\sum_{j=1}^k w_{e_{i_j}})$, then it includes the element $e_i$ and discards the elements $\cup_{j=1}^k \{e_{i_j}\}$. We will prove that the above algorithm is $\frac{(k\cdot r-1)r}{r-1-f}$ competitive. Here $r$ is a constant defined later to optimize the competitive ratio.

The Algorithm 2 in Figure 1 is an equivalent formulation of Algorithm 1. It is not difficult to show that steps 6-12 of Algorithm 1 are equivalent to step 6 of Algorithm 2. For ease of analysis we will just analyze Algorithm 1. 

\section{Analysis of the algorithm} \label{s:analysis}
Let $S(i)$ be the set $S$ at the end of step $i$ and let $OPT\subseteq \mathcal{E}$ be optimal solution to the weighted intersection of $k$ matroids. The main part of competitive analysis is based on the following lemma. 
\begin{lemma} \label{mainlemma}
$w(S(n)) \frac{(k\cdot r-1)r}{r-1}\geq w(OPT)$ where $w(S(n))=\sum_{e\in S(n)}w_e$.
\end{lemma}
We will prove Lemma \ref{mainlemma} in \ref{chargeSection}. The proof is based on a new type of charging scheme. For now we just assume it to analyze the competitive ratio of the algorithm. We once again note that the main technical contribution is in analyzing Lemma \ref{mainlemma} and given Lemma \ref{mainlemma} the analysis of Theorem \ref{theoremonenew}(i.e bounding the penalty) is very similar to \cite{BHK-ec09}.

\begin{theorem} \label{theoremonenew}
The online algorithm with cancellations for $k$ matroid constraints has a competitive ratio $c=\frac{(k\cdot r -1)r}{r-1-f}$. This ratio is minimized when $\frac{r}{1+f}=1+\sqrt{1-\frac{1}{k(1+f)}}$ and has a value $c=k(1+f)(1+\sqrt{1-\frac{1}{k(1+f)}})^2$
\end{theorem}
\begin{proof} The competitive ratio of our algorithm matches the case $k=1$ given in \cite{BHK-ec09,CFMP-soda09}. Later in section \ref{s:lowerbound} we will show that this is tight for every $k$.

The utility of the algorithm comprises of two terms. One is due to utility of $S(n)$ and the other is the penalty due to canceled set $R$. 
\begin{itemize}
\item For each element $e_i$ we define a value $P(e_i)$ recursively. If $e_i$ was accepted in step 3 or was never accepted, then $P(e_i)=0$. Else if elements $C_{e_i}=\{e_{i_1},e_{i_2},\ldots,e_{i_k}\}$ were canceled, then $P(e_i)=f\cdot \sum_{j=1}^k w_{e_{i_j}}+\sum_{j=1}^k P(e_{i_j})$. Now each canceled item $e_{i_j}$'s penalty is accounted to the item $e_i$ which canceled it. We prove that for any element $e_i$ the total penalty accounted is less than or equal to $\frac{f}{r-1}\cdot w_{e_i}$. The proof is by induction. The base case when $P(e_i)=0$ is simple. The inductive case is as follows.
\begin{eqnarray}
P(e_i)=& f\cdot \sum_{j=1}^k w_{e_{i_j}}+\sum_{j=1}^k P(e_{i_j}) \nonumber \\
\leq& f\cdot \sum_{j=1}^k w_{e_{i_j}}+\sum_{j=1}^k f\cdot \frac{w_{e_{i_j}}}{r-1} \nonumber \\
=& \frac{fr}{r-1}\cdot \sum_{j=1}^k w_{e_{i_j}} \nonumber \\
\leq& \frac{fr}{r-1} \frac{w_{e_i}}{r} \nonumber \\
=& f\cdot \frac{w_{e_i}}{r-1}
\end{eqnarray}
Hence the total penalty is at-most $\sum_{e_i\in S(n)}f\cdot \frac{w_{e_i}}{r-1}=\frac{f\cdot w(S(n))}{r-1}$
%\item For each element $e_i$ we define a set $B$ recursively. If $e_i$ was accepted in step 3 or was never accepted then $B(e_i)=\pthi$. Else if elements $C_{e_i}=\{e_{i_1},e_{i_2},\ldots,e_{i_k}\}$ were rejected then $B(e_i)=\{f\cdot \sum_{j=1}^k w_{e_{i_j}}\}\cup (\cup_{j=1}^k B(e_{i_j}))$. Now each canceled item $e_{i_j}$'s penalty is accounted to the item $e_i$ which canceled it. We prove that for any element $e_i$ the total penalty accounted is less than or equal to $\frac{f}{r-1} w_{e_i}$. The proof is by induction. The base case when $B(e_i)=\pthi$ is simple. The inductive case is as follows.
%\begin{eqnarray}
%Penalty\ due\ to\ e_i\leq& f\cdot \sum_{j=1}^k w_{e_{i_j}}+\sum_{j=1}^k Penalty\ due\ to\ e_{i_j} \nonumber \\
%\leq& f\cdot \sum_{j=1}^k w_{e_{i_j}}+\sum_{j=1}^k f\cdot \frac{w_{e_{i_j}}}{r-1} \nonumber \\
%\leq& \frac{f\cdot r}{r-1}\cdot \sum_{j=1}^k w_{e_{i_j}} \nonumber \\
%\leq& \frac{f\cdot r}{r-1} \frac{w_{e_i}}{r} \nonumber \\
%\leq& f\cdot \frac{w_{e_i}}{r-1}
%\end{eqnarray}
%Hence the total penalty is at-most $\sum_{e_i\in S(n)}f\cdot \frac{w_{e_i}}{r-1}=\frac{f\cdot w(S(n))}{r-1}$
\item The final weight of the set $S(n)$ is bounded by Lemma \ref{mainlemma}. Combining the two parts we get the total utility of the algorithm.
\begin{eqnarray}
\textrm{Utility}\geq& w(S(n))-f\cdot \frac{w(S(n))}{r-1} \nonumber \\
=& \frac{r-1-f}{r-1}\cdot w(S(n)) \nonumber \\
\geq& \frac{r-1-f}{(k\cdot r-1) r}\cdot w(OPT) &\textrm{(Using Lemma \ref{mainlemma})}
\end{eqnarray}
\end{itemize}
Hence we get competitive ratio of $c=\frac{(k\cdot r-1)r}{r-1-f}$. Optimizing over $r$ we get $\frac{r}{1+f}=1+\sqrt{1-\frac{1}{k(1+f)}}$ and $c=k(1+f)(1+\sqrt{1-\frac{1}{k(1+f)}})^2$.
\end{proof}

\subsection{Charging scheme} \label{chargeSection}
Here we will prove Lemma \ref{mainlemma}. This portion is technically the hardest part of the paper and requires developing a new charging scheme. This is aided by a graph construction. We first give some notation. Each element carries two kinds of charges($\mathsf{ch_1}$ and $\mathsf{ch_2}$). Let at any step $j$, $\mathsf{ch}(e_i,j)=\mathsf{ch_1}(e_i,j)+\mathsf{ch_2}(e_i,j)$. Additionally let $C_{e_i}=\cup_{j=1}^k e_{i_j}$(the set of elements discarded when $e_i$ is included) be as defined in step 13 of Algorithm 1. Let $S(j)$ denote the set $S$ after step $j$. Let $\mathsf{ch}(S',j)=\sum_{e\in S'} \mathsf{ch}(e,j)$(analogously for $\mathsf{ch_1}$ and $\mathsf{ch_2}$).

\subsubsection{Sketch}
We start with a total charge of $OPT$ on the elements. $\mathsf{ch_1}(e_i,j)$ denotes the charge which the element carries from the beginning. $\mathsf{ch_2}(e_i,j)$ denotes the charge which the elements gets from some other element. At any step either $e_i$ or $C_{e_i}$ is discarded. When any element is discarded all its charge is re-added to $S(i)$(or $S(j)$ for $j\geq i$). There by all the charge is stored in $S(n)$. Next we bound the amount of charge any element can carry.

There are two ways charge is transferred. One way is for $\mathsf{ch_1}$ and another for $\mathsf{ch_2}$. At step $i$ if $C_{e_i}$ is discarded, then we always transfer $\mathsf{ch_2}(C_{e_i},i-1)$ to $\mathsf{ch_2}(e_i,i)$. Another way of transfer is for $\mathsf{ch_1}$. This is done by a fairly sophisticated graph construction. Essentially we construct $k$ bipartite graphs and then prove that each one has a matching that matches all vertices on the left side of the bipartition using Hall's Theorem. Suppose $e_i$ is matched to $e_j$, we transfer $\mathsf{ch_1}(e_i,a)$ to $\mathsf{ch_2}(e_j,b)$. Here $a$ denotes the step at which $e_i$ was removed and $b$ denotes the step before which $e_j$ was removed($n$ otherwise). Note that we will need causal consistency($a<b$) for this.

The charges at the beginning of the algorithm are defined as follows.
\begin{itemize}
\item if $e_i\in OPT$ then $\mathsf{ch}(e_i,0)=\mathsf{ch_1}(e_i,0)=w_{e_i}$ and $\mathsf{ch_2}(e_i,0)=0$
\item if $e_i\notin OPT$ then $\mathsf{ch}(e_i,0)=\mathsf{ch_1}(e_i,0)=\mathsf{ch_2}(e_i,0)=0$.
\end{itemize}
Before defining the charging scheme we define a graph construction which will aid us in the charging scheme. 
\subsubsection{Graph Construction}
Construct $k$ bi-partite graphs as the algorithm proceeds. Here $p^{th}$ graph corresponds to the $p^{th}$ matroid. Let $P_1(p)$ denote partite set 1 of $p^{th}$ graph and $P_2(p)$ denote partite set 2 of $p^{th}$ graph. Additionally let $N_p(S)$ denote the set of neighbors of $S\subseteq P_1(p)$ in $p^{th}$ graph. Let $rank_p(S)$ be the rank of set $S$ in the $p^{th}$ matroid.
\begin{enumerate}
\item The bi-partite graph starts empty and edges are added. Each end point of an edge corresponds to an element $e_i$. The node corresponding to an element $e_i$ exists only when corresponding edge is added and removed when all its adjacent edges are deleted. An edge in the graph corresponds to a potential $\mathsf{ch_1}$ transfer.
\item \label{delone} Consider step 14 in the algorithm. If $w_{e_i}<r\cdot \sum_{j=1}^k w_{e_{i_j}}$, then $e_i$ is not included in $S$. Now if $e_i\in OPT$, then add a node $e_i$ to $P_1(p)$ (for each $p$). Let $Ckt(e_i,p)$ be the unique circuit in $p^{th}$ matroid in $S\cup \{e_i\}$. Then add edge $e_i,e_j$ for each $e_j\in Ckt(e_i,p)-\{e_i\}$ with $e_j$ belonging to $P_2(p)$.
\item \label{deltwo} Consider step 14 in the algorithm. If $w_{e_i}\geq r\cdot \sum_{j=1}^k w_{e_{i_j}}$, then $C_i=\{e_{i_1},\ldots,e_{i_k}\}$ is deleted from $S$ and $e_i$ is included into it. Delete each $e_{i_p}$ from the corresponding $P_2(p)$\footnote{Note that $e_{i_p}$ is deleted only from $P_2(p)$ and not from $P_2(p')$ for $p'\neq p$}. For each existing edge $e_{q},e_{i_p}$ add edges $e_q,e_j$ for each $e_j\in Ckt(e_i,p)-\{e_{i_p}\}$ with $e_j$ belonging to $P_2(p)$. Additionally if $e_{i_p}\in OPT$(i.e $\mathsf{ch_1}(e_{i_p},0)>0$), then re-add it to $P_1(p)$. Add edges $e_{i_p},e_j$ for each $e_j\in Ckt(e_i,p)-\{e_{i_p}\}$.
\end{enumerate}
\begin{lemma} \label{properties}
The graph construction has the following properties.
\begin{enumerate}
\item $P_1(p)\subseteq OPT-S(n)$ for each $p$.
\item \label{span} $\forall \hat{S}\subseteq P_1(p)$,$N_p(\hat{S})$ spans $\hat{S}\subseteq P_1(p)$ in $p^{th}$ matroid.
\item \label{number} $\forall \hat{S}\subseteq P_1(p)$,$|N_p(\hat{S})-OPT|\geq|\hat{S}|$.
\item There exists a matching in graph $p$ such that every $e\in P_1(p)$ is matched
to a node in $P_2(p)-OPT$.
\item \label{prop:four} Any element $e\in \mathcal{E}-S(n)-OPT$ is matched in at-most $k-1$ of the graphs from the side of $P_2$.
\end{enumerate}
\end{lemma}
\begin{proof}
\begin{enumerate}
\item This is easily seen by construction. In steps \ref{delone} and \ref{deltwo} of Graph construction an element is added to $P_1(p)$ precisely when it is removed from $S$ and when it belongs to $OPT$. 
\item When an node $e_j$ is added to $P_1(p)$ then edges to each element in $Ckt-\{e_j\}$ are added. Hence
any node $e_j$ is spanned by $N_p(\{e_j\})$. By matroid property this implies that for any set $\hat{S}\subseteq P_1(p)$ we have that $N_p(\hat{S})$ spans $\hat{S}$. \footnote{Note that even though the edges could later be deleted, the span property still holds due to additional edges being added.}
\item Let $W=N_p(\hat{S})\cap OPT$. We now assert some statements from which the inequality easily follows.
\begin{itemize}
\item $rank_p((N_p(\hat{S})-OPT)\cup W)\geq rank_p(\hat{S}\cup W)$\footnote{$rank_p(S)$ is defined as the largest subset $A\subseteq S$ such that $A\in \mathcal{I}_p$}. This follows from property \ref{span}.
\item $rank_p(\hat{S}\cup W)=rank_p(\hat{S})+rank_p(W)=|\hat{S}|+|W|$. This follows from the fact that $\hat{S}$ and $W$ are disjoint and $\hat{S}\cup W\subseteq OPT$.
\item $rank_p((N_p(\hat{S})-OPT)\cup W)\leq rank_p(N_p(\hat{S})-OPT)+rank_p(W)\leq |N_p(\hat{S})-OPT|+|W|$. These set of inequalities follows from the matroid property.
\end{itemize}
Combining the above equations we get $|\hat{S}|\leq |N_p(\hat{S})-OPT|$.
%Then from property \ref{span} we have that $rank_p((N_p(\hat{S})-OPT)\cup W)\geq rank_p(\hat{S}\cup W)$. Additionally from construction we have that $\hat{S}$ and $W$ are disjoint and $\hat{S}\cup W\subseteq OPT$. Hence $rank_p(\hat{S}\cup W)=rank_p(\hat{S})+rank_p(W)=|\hat{S}|+|W|$. We also know from matroid property that $rank_p((N_p(\hat{S})-OPT)\cup W)\leq rank_p(N_p(\hat{S})-OPT)+rank_p(W)\leq |N_p(\hat{S})-OPT|+|W|$. Combining the equations we get $|\hat{S}|\leq |N_p(\hat{S})-OPT|$
\item Follows from Hall's Theorem and property \ref{number}. If $e_i\in P_1(p)$ is matched to $e_j$, then let $e_j=M_p(e_i)$.
\item This follows from step \ref{deltwo} of graph construction. Here any element removed from $S(i)$ in any step is deleted from $P_2(p)$(for one of the $p$'s). Hence such an element could belong to $P_2$ of at-most $k-1$ graphs and be matched at-most $k-1$ times(from $P_2$'s side). Note that any element $e\in OPT-S(n)$ could additionally be matched from the $P_1$'s side.
\end{enumerate}
\end{proof}
\subsubsection{Charge Transfer}
%As explained in the sketch there are two ways a charge could be transfered.
%\begin{itemize}
%\item The transfer of $\mathsf{ch_2}$ is quite simple. In step 15 of the algorithm when $C_{e_i}$ is discarded at step $j$, transfer $\mathsf{ch_2}(C_{e_i},j-1)$ to $\mathsf{ch_2}(e_i,j)$.
%\item The transfer of $\mathsf{ch_1}$ is quite complicated and uses the graph construction in the previous section. Note that a $\mathsf{ch_1}$ transfer must happen only when a $e\in OPT$ is discarded. 
%\end{itemize}
We finally explain the exact way the charge is transferred in each step.
\begin{itemize}
\item Consider step 14 in the algorithm. If $w_{e_i}\geq r\cdot (\sum_{j=1}^k w_{e_{i_j}})$, then transfer all of $\mathsf{ch_2}(e_{i_j},i-1)$ for each $e_{i_j}$ to $\mathsf{ch_2}(e_i,i)$.
\item Consider step 14 in the algorithm. If $w_{e_i}\geq r\cdot (\sum_{j=1}^k w_{e_{i_j}})$, then let $C_{e_i}=\{e_{i_1},\ldots,e_{i_k}\}$ be deleted from $S$ and $e_i$ is included into it. If $e_{i_l}$ was re-added to $P_1(l)$, then transfer all of $\mathsf{ch_1}(e_{i_l},0)$ to $\mathsf{ch_2}(M_l(e_{i_l}),t)$ (where $t$ is either the step $M_l(e_{i_l})$ is deleted or n).
\item Consider step 14 in the algorithm. Let $w_{e_i}<r\cdot (\sum_{j=1}^k w_{e_{i_j}})$. Additionally if $e_i\in OPT$, then it would have been added to $P_1(l)$ of each graph $l$. Now $e_i$ is matched to different nodes in different graphs. Transfer a portion of $\mathsf{ch_1}(e_i,0)$ to $\mathsf{ch_2}(M_l(e_i),t)$ which is proportional to $w_{e_{i_l}}$ for each graph $l$. (where $t$ is either the step $M_l(e_i)$ is deleted or n).
\item Note that the above transfer of charges does not violate causal consistency as the transfer of charge happens from $e$ to some element in $S(i)$ of the future.
\end{itemize}
We finish the proof of Lemma \ref{mainlemma} by analyzing the charge transfer. First note that any element in $\mathcal{E}-S(n)$ receives $\mathsf{ch_1}$ transfer in step 2 or 3 of charge transfer at-most $k-1$ times. This is by property \ref{prop:four} of Lemma \ref{properties}. Additionally we can also see that each $\mathsf{ch_1}$ transfer to element $e_i$ is at-most $r\cdot w_{e_i}$. Using these properties by Induction that we prove that $\mathsf{ch_2}(e_i,j)\leq \frac{(k-1)r^2}{r-1} w_{e_i}$ for $e_i\in \mathcal{E}-S(n)$ and $\mathsf{ch}(e_i,j)\leq \frac{(k\cdot r-1)r}{r-1} w_{e_i}$ for $e_i\in S(n)$.
\begin{itemize}
\item For any element in $e_i\in \mathcal{E}-S(n)$ we have $\mathsf{ch_2}(e_i,j)\leq \frac{(k-1)r^2}{r-1} w_{e_i}$ if it was deleted at step $j$. Note that $\mathsf{ch_1}$ transfer happens atmost $k-1$ times for $e_i\in \mathcal{E}-S(n)-OPT$ and $0$ times for $e_i\in \mathcal{E}-S(n)\cap OPT$
\begin{eqnarray}
\mathsf{ch_2}(e_i,j)\leq& (k-1)\cdot (\mathsf{ch_1}\ \textrm{transfer})+\mathsf{ch_2}\ \textrm{transfer} \nonumber \\
\leq& (k-1)\cdot r\cdot w_{e_i}+\sum_{j=1}^k \mathsf{ch_2}(e_{i_j},i-1) \nonumber \\
\leq& (k-1)\cdot r\cdot w_{e_i}+\sum_{j=1}^k \frac{(k-1)r^2}{r-1} w_{e_{i_j}} \nonumber \\
=& (k-1)\cdot r\cdot w_{e_i}+\frac{(k-1)r^2}{r-1}\sum_{j=1}^k  w_{e_{i_j}} \nonumber \\
\leq& (k-1)\cdot r\cdot w_{e_i}+\frac{(k-1)r^2}{r-1}\frac{w_{e_i}}{r} \nonumber \\
=& \frac{(k-1)r^2}{r-1} w_{e_i}
\end{eqnarray}
\item For any element in $S(n)-OPT$ we have $\mathsf{ch}(e_i,n)\leq \frac{(k\cdot r-1)r}{r-1} w_{e_i}$. Note that this is also true for any element in $S(n)\cap OPT$, as they do not get any $\mathsf{ch_1}$ transfer but have a non-zero $\mathsf{ch_1}$.
\begin{eqnarray}
\mathsf{ch}(e_i,n)\leq& k\cdot (\mathsf{ch_1}\ \textrm{transfer})+\mathsf{ch_2}\ \textrm{transfer} \nonumber \\
\leq& k\cdot r\cdot w_{e_i}+\sum_{j=1}^k \mathsf{ch_2}(e_{i_j},i-1) \nonumber \\
\leq& k\cdot r\cdot w_{e_i}+\sum_{j=1}^k \frac{(k-1)r^2}{r-1} w_{e_{i_j}} \nonumber \\
=& k\cdot r\cdot w_{e_i}+\frac{(k-1)r^2}{r-1}\sum_{j=1}^k  w_{e_{i_j}} \nonumber \\
\leq& k\cdot r\cdot w_{e_i}+\frac{(k-1)r^2}{r-1}\frac{w_{e_i}}{r} \nonumber \\
=& \frac{(k\cdot r-1)r}{r-1} w_{e_i}
\end{eqnarray}
\end{itemize}
The above argument proves that 
\begin{itemize}
\item $\forall e_i\in S(n)$, $\mathsf{ch}(e_i,n)\leq \frac{(k\cdot r-1)r}{r-1} w(e_i)$
\item $\sum_{e_i\in S(n)} \mathsf{ch}(e_i,n)=w(OPT)$. This follows naturally from charge conservation in the system.
\end{itemize}
The proof of Lemma \ref{mainlemma} can be easily seen from the above two properties.

\section{Lower Bound}\label{s:lowerbound}
\subsubsection{Sketch} Here we prove a matching lower bound of $c=k(1+f)(1+\sqrt{1-\frac{1}{k(1+f)}})^2$ for the competitive ratio. Assume $\mathcal{A}$ is an online deterministic algorithm which achieves a competitive ratio $\beta< c$. Then we arrive at a contradiction. The proof will be in following steps.
\begin{enumerate}
\item We will construct a k-dimensional matching. Using this we will argue the existence of an infinite sequence $X=\{x_1,x_2,\ldots\}$ of the following form. $x_1=1$ and $x_i>0,\forall i$. Additionally they will satisfy the following inequality.
\begin{equation} \label{ineq}
\beta (x_i-f\cdot \sum_{j=1}^{i-1}x_j)\geq x_{i+1}+(k-1)\sum_{j=1}^{i+1} x_j, \forall i\geq 1
\end{equation}
\item Consider any sequence $X=\{x_1,x_2,\ldots\}$ which satisfies $x_i>0,\forall i$ and $\beta (x_i-f\cdot \sum_{j=1}^{i-1}x_j)\geq x_{i+1}+(k-1)\sum_{j=1}^{i+1} x_j,\forall i$. Now if $\beta< c$, we arrive at a contradiction. 
\end{enumerate}
We once again note that the important new idea here is in first step. The first step needs to get the exact matroid structures right so as to pin the tight inequality. The lower bound given in \cite{BHK-ec09} for $k=1$ is just a uniform matroid of rank 1 while for $k>1$ we need intersection of partition matroids which has a lot more structure. Given the part 1 in our construction part 2 follows more easily from techniques in \cite{BHK-ec09}. We will prove the second part first.

\subsection{Contradiction}
Consider all sequences of the form $x_1=1$,$x_i>0,\forall i$ and satisfying equation \ref{ineq}. We first claim that if a sequence of the above form exists, then there should exist a sequence where inequality \ref{ineq} is equality $\forall i$. Assume by contradiction that such a sequence does not exist. For any given sequence X, let $n(X)$ be the minimum $i$ for which inequality \ref{ineq} is strict. Among the sequences consider sequence $X$ for which $n(X)=N$ is as large as possible. We construct a sequence for which $n(X)$ is even larger thus arriving at a contradiction. Let $\lambda$ be defined as follows.
\begin{equation}
\lambda=\frac{\beta (x_N-f\cdot \sum_{j=1}^{N-1}x_j)-(k-1)\sum_{j=1}^{N} x_j}{k\cdot x_{N+1}}
\end{equation}
Let $X'=\{x_1,x_2,\ldots,x_N,\lambda x_{N+1},\lambda_{N+2}x_{N+2},\ldots\}$. Then it is easy to see that inequality \ref{ineq} and other constraints are met. Additionally $n(X')>n(X)$. Hence we arrive at a contradiction to the non-existence of a sequence when inequality \ref{ineq} is equality $\forall i$. Let $Y=\{y_1,y_2,\ldots\}$($y_i\geq 0$) be the sequence such that 
\begin{eqnarray}
y_1=1,\ \ \ \ \ \ \ \ \ \ \  \beta (y_i-f\cdot \sum_{j=1}^{i-1}y_j)= y_{i+1}+(k-1)\sum_{j=1}^{i+1} y_j,\forall i\geq 1
\end{eqnarray}
Let $z_i=\sum_{j=1}^i y_i$. Then we get the recurrence $k\cdot z_{i+1}=(1+\beta)z_i-\beta (1+f)z_{i-1}$ and $z_1=1$. As each $y_i\geq 0$ we also have that $z_i\geq 0$. Now we use a Lemma from \cite{Positivity} to get a contradiction. 

\begin{lemma} \cite{Positivity} \label{contra}
Let $u_n=au_{n-1}+bu_{n-2}$ be a linear recurrence of second order. If $p(x)=x^2-ax-b$ has imaginary roots, then the sequence must have negative elements.
\end{lemma} 

Consider the case when $\beta< k(1+f)(1+\sqrt{1-\frac{1}{k(1+f)}})^2$. Then it is simple to see that $D=(1+\beta)^2-4\cdot k\cdot \beta(1+f))<0$ which implies that $k\cdot x^2=(1+\beta)x-\beta(1+f)$ has imaginary roots. Hence by Lemma \ref{contra} this implies that the sequence $\{z_i\}$ has negative elements which is a contradiction to our assumption.

\subsection{Construction of sequence}
\begin{figure}
\begin{center}

\input{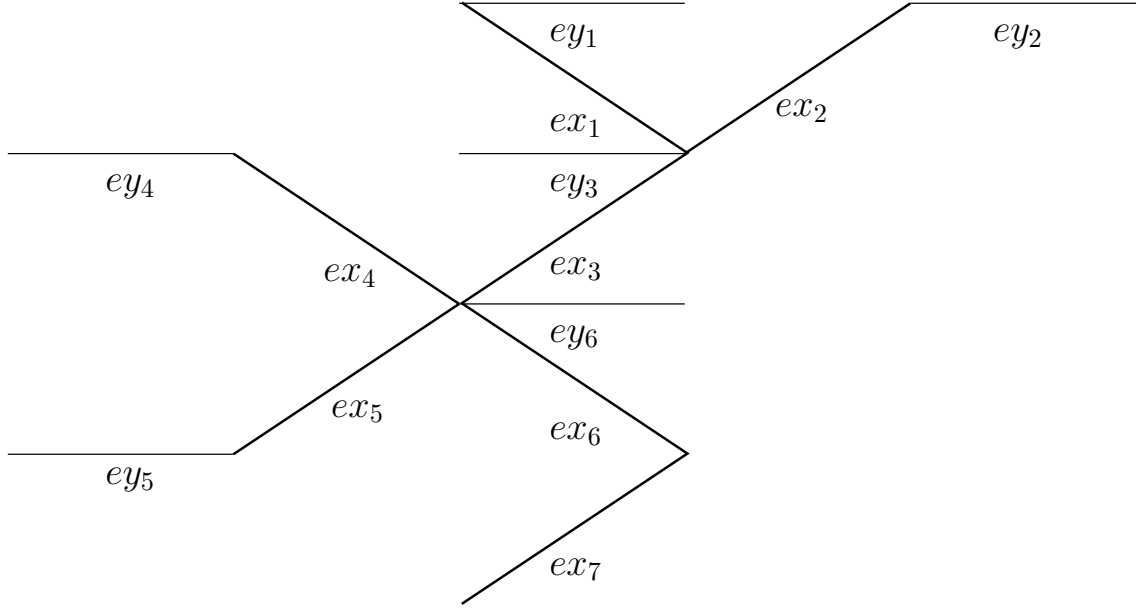}

\end{center}
\caption{\label{example}Example $k$ dimensional matching. Step before adding $ey_7$,$ex_8$}
\end{figure}

Given the online algorithm with competitive ratio $\beta$ we construct $k$ dimensional matching using which we will construct the sequence $\{x_i\}$. For ease of description we will restrict to the case $k=2$. The general $k$ case is similar. In other words we construct bi-partite graphs. Here each partite corresponds to a partition matroid. Hence the subset of edges is in the intersection of the two matroids if and only it forms a valid matching. The graph will have two types of edges $ex_i$ and $ey_i$. We will use $x_i$ and $y_i$ to denote the corresponding weights of the edges. 
\begin{itemize}
\item Start with edge $ex_1$ with weight $x_1=1$.
\item For simplicity we just state the inductive step in the construction. At step $i$ the edge held by the algorithm is $ex_i$(and no other edge).
\item At step $i+1$ we add edges to both ends of $ex_i$ such that they differ in weight by atmost $\epsilon$ and the algorithm accepts exactly 1 of them. Due to the matching condition it has to reject the currently accepted edge $ex_i$. Here adding a new edge to one end of $ex_i$ means a edge is revealed. This new edge shares one vertex with $ex_i$ and the other vertex is a brand new vertex which hasn't yet appeared in the algorithm. We describe how this is done below.

At step $i+1$ add edges to both ends of $ex_i$ of weight $\epsilon$. If the algorithm does not accept either of the new edges rewind the algorithm and instead add edges of weight $2\epsilon$. Do this rewind, add edges of higher weights(higher by $\epsilon$) till the algorithm accepts exactly 1 new edge. Due to the matching constraint this means the current edge $ex_i$ sharing an end point must be canceled and a penalty paid to it. Let the accepted edge be named $ex_{i+1}$ and the other newly added edge be named $ey_i$. By construction we have $x_{i+1}\leq y_i+\epsilon$. 
\item Use the above construction to construct an infinite sequence $x_i$ and $y_i$. An example construction is shown in figure \ref{example}.

\end{itemize}
We will note some properties of the sequence $x_i$ and $y_i$.
\begin{itemize}
\item At step $i+1$ the algorithm accepts $ex_{i+1}$ and cancels $ex_i$ while not accepting $ey_i$. Consider the rewinding procedure in which the algorithm is presented edge $e'$ of weight $x_{i+1}-\epsilon$ and $ey_i$ before $ex_{i+1}$ is presented. By the construction of our rewinding procedure both $e'$ and $ey_i$ would not be accepted and $ex_i$ would still be the currently maintained edge in the solution. In such a case we assert some statements based on which we derive an inequality.
\begin{itemize}
\item The utility of the algorithm is $x_i-f\sum_{j=1}^{i-1}x_{i-1}$ as $ex_i$ is the currently maintained edge and $\{ex_1,...,ex_{i-1}\}$ are currently canceled edges.
\item It is clear that $\{ey_1,ey_2,\ldots,ey_{i},e'\}$ is a valid matching in the current set of revealed edges. This has total weight $x_{i+1}-\epsilon+\sum_{j=1}^i y_i$
\item By definition of $\beta$ we have $\beta (x_i-f\sum_{j=1}^{i-1}x_{i-1})\geq x_{i+1}-\epsilon+\sum_{j=1}^i y_i$.
\item By construction of the sequence we also know that $y_i+\epsilon\geq x_{i+1}$.
\end{itemize}
Substituting we get $\beta (x_i-f\sum_{j=1}^{i-1}x_{i-1})\geq x_{i+1}-\epsilon+\sum_{j=2}^{i+1} x_{j}-(i+1)\epsilon$. Tending $\epsilon$ to $0$ \footnote{We are ignoring some issues of convergence for ease of exposition} we get $\beta (x_i-f\sum_{j=1}^{i-1}x_{i-1})\geq x_{i+1}+\sum_{j=2}^{i+1} x_{j}$
\item Note that the sequence constructed in not the one we desired. The sum on the right hand side starts from $j=2$. $\beta (x_i-f\sum_{j=1}^{i-1}x_{i-1})\geq x_{i+1}+\sum_{j=2}^{i+1} x_{j}$ implies $\beta (x_i-f\sum_{j=2}^{i-1}x_{i-1})\geq x_{i+1}+\sum_{j=2}^{i+1} x_{j}$. So from the sequence $\{x_1,x_2,...\}$ deleting $x_1$ and rescaling $x_2$ to 1 we get the desired sequence.
\end{itemize}
The general case k is very similar to case $k=2$ but instead involves construction of $k$-dimensional matching.

\section{Extending to arbitrary downward closed set systems.}\label{s:extensions}
We extend the algorithm in section \ref{s:preliminaries} to arbitrary downward closed set systems. One way is to represent the set system as a intersection of $k$ matroids and use the algorithm in previous sections. But even the algorithm for single item from \cite{BHK-ec09} gives a competitive ratio $n\cdot (1+2f+2\sqrt{f(1+f)})$. A simple modification can be used to improve it to $(n-1)(1+2f+2\sqrt{f(1+f)})$. As we show next this is essentially the best that can be done.
\begin{theorem}
Even for the case $f=0$ the competitive ratio cannot be better than $n-1$.
\end{theorem}
\begin{proof}
The proof is by constructing a set system for which every algorithm achieves competitive ratio worse then $n-1$. Consider the downward closed set system $\mathcal{I}$ which is the set of independent sets in a graph. Now we construct the graph as follows. At every step a node in the graph arrives and reveals its weight and edges to existing set of vertices. First node $N_1$ and $N_2$ with weights $1$ arrive and have a edge between them. Without loss of generality assume that the algorithm $\mathcal{A}$ accepts $N_1$ and rejects $N_2$. Next at every step a node $N_i$ arrives with weight $1-\epsilon$ and edge to $N_1$. The algorithm can never reject $N_1$ and accept a new $N_i$ without violating the fact that its competitive ratio if less than $n-1$. But at the end $OPT=\{N_2,N_3,\ldots,N_n\}$ with weight $1+(n-2)(1-\epsilon)$ while the weight held by algorithm is $1$. Hence the competitive ratio is at-least $1+(n-2)(1-\epsilon)$. Tending $\epsilon->0$ we get that competitive ratio is at-least $n-1$. 

%The above proof can be generalized to prove that for every $f$ the competitive ratio is at-least $(n-O(1))(1+2f+2\sqrt{f(1+f)}-O(1))$.
\end{proof}

\section{Acknowledgments}
The author thanks Robert D. Kleinberg, Renato Paes Leme and Hu Fu for useful comments and suggestions on an earlier draft of the paper.

{\footnotesize
\bibliographystyle{splncs}
\bibliography{buyback}
}
%\appendix
%\input{appendix}
\end{document}